\documentclass[runningheads]{llncs}
\usepackage{cite}
\usepackage{amsmath,amssymb,amsfonts}
\usepackage{algorithmic}
\usepackage{graphicx}
\usepackage{textcomp}
\usepackage{centernot}
\usepackage[outline]{contour}
\usepackage[at]{easylist}
\usepackage[table,xcdraw]{xcolor}
\usepackage{caption}
\usepackage{subcaption}
\usepackage{multicol}
\usepackage{amsmath}
\usepackage{pgf,tikz}
\usepackage{wrapfig}
\usepackage{capt-of}
\usetikzlibrary{arrows,shapes}
\usepackage{booktabs}

\spnewtheorem{assumption}[theorem]{Assumption}{\bfseries}{\itshape}
\spnewtheorem{notation}[theorem]{Notation}{\bfseries}{\itshape}
\newcommand{\qedhere}{\qed}

\usepackage{apptools}
\AtAppendix{\counterwithin{lemma}{section}}

\usepackage{todonotes}

\tikzstyle{vertex}=[circle,fill=black!10,minimum size=15pt,inner sep=0pt]
\tikzstyle{selected vertex} = [vertex, fill=black!60, text=white]
\tikzstyle{arrow} = [draw,thick,->]
\tikzstyle{weight} = [font=\small]
\tikzstyle{selected edge} = [draw,line width=5pt,-,red!50]
\tikzstyle{ignored edge} = [draw,line width=5pt,-,black!20]
\definecolor{ruddypink}{rgb}{0.88, 0.56, 0.59}
\definecolor{aquamarine}{rgb}{0.5, 1.0, 0.83}

\newcommand{\source}{\mathit{source}}
\newcommand{\target}{\mathit{target}}
\newcommand{\route}{\mathit{rout}}
\newcommand{\nexthop}[1]{\mathit{next}_{#1}}
\newcommand{\nextC}[1]{\mathit{nextC}_{#1}}
\newcommand{\AP}{\mathit{AP}}

\newcommand{\N}{\mathcal{N}}
\newcommand{\KSN}{K_{\N}}

\begin{document}

\title{Deadlock in packet switching networks}

\author{Anna Stramaglia \and
Jeroen J.A. Keiren \and
Hans Zantema}

\institute{Eindhoven University of Technology, the Netherlands\\ 
\email{\{a.stramaglia, j.j.a.keiren, h.zantema\}@tue.nl}}

\maketitle 
\begin{abstract}
A deadlock in a packet switching network is a state in which one or more messages have not yet reached their target, yet cannot progress any further.
We formalize three different notions of deadlock in the context of packet switching networks, to which we refer as global, local and weak deadlock.
We establish the precise relations between these notions, and prove they characterize different sets of deadlocks.
Moreover, we implement checking of deadlock freedom of packet switching networks using the symbolic model checker nuXmv.
We show experimentally that the implementation is effective at finding subtle deadlock situations in packet switching networks.

\keywords{Packet switching network \and Deadlock \and Model checking.}
\end{abstract}

\section{Introduction}\label{section:introduction}
Deadlock is a historically well known bug pattern in computer systems where, in the most general sense, a system reaches a state in which no operation can progress any further.
Deadlocks can occur in many different contexts, such as operating systems~\cite{CES1974}, databases~\cite{Wol1986}, computer networks, and many others~\cite{Zob1983},
provided one interprets the processes and resources involved appropriately. Regardless of the context,  deadlock is a situation that we generally want to avoid.


A packet switching network consists of nodes, connected by (directed) channels. Packets are exchanged in a store-and-forward manner. This means that a node in the network first receives a packet in its entirety, and then decides along which output channel to forward the packet based on a routing function. The possible steps in the network are: sending a packet to some other node, processing the packet by first receiving and then forwarding it, and finally, receiving a packet when it reaches its destination node.

Packet switching networks have been around for decades, and the problem of deadlock in such networks was already described early on~\cite{MS1980}. Basically a deadlock arises if packets compete for available channels. There are different ways to deal with deadlocks. First, in deadlock avoidance, extra information in the network is used to dynamically ensure deadlock freedom. Second, in deadlock prevention, deadlock freedom is ensured statically, e.g. based on the network topology and the routing function. Finally, networks with deadlock detection are less restrictive in their routing. Deadlocks that result from these relaxed routing schemes are detected and resolved using an online algorithm~\cite{Chen1974,Lopez2011}.

Many packet switching networks have a dynamic topology, and therefore use deadlock avoidance or deadlock detection. However, from the early 2000s, Networks on Chip (NoCs) brought packet switching and deadlock prevention to the level of interconnect networks in integrated circuits~\cite{BM2002,DT2004}.
Since such NoCs have a static topology, they are amenable to deadlock prevention.

Deadlock prevention was studied, e.g., by Chen in 1974~\cite{Chen1974}, who referred to prevention as \emph{``system designs with built-in constraints which guarantee freedom from deadlocks without imposing any constraints on real-time resource allocation''}. Later, in the 1980s, Toueg and Ullman addressed deadlock prevention using local controllers~\cite{TU1981}.
Duato~\cite{Dua1996} was the first one to propose necessary and sufficient conditions for deadlock-free routing in packet switching.
In the context of NoCs, Verbeek~\cite{VS2009,VS2010,Ver2013} formulated a necessary and sufficient condition for deadlock-free routing that is equivalent to that of Duato.
The notion of \emph{local deadlock} we introduce in Section~\ref{subsection:local-deadlock} is equivalent to those of Duato and Verbeek.
This paper is based on preliminary results in~\cite{Str2020}.

\paragraph{Contributions}
In this paper, we focus on \emph{deadlock prevention} in packet switching networks, with a particular interest in NoCs.
We restrict ourselves to networks with deterministic, incremental and node-based routing functions. We formalize three different notions of deadlock, namely global, local and weak deadlock.
The definition of global deadlock is the standard definition in which no message can make progress in the entire network.
A weak deadlock is a state in which no steps other than send steps are possible. A state is a local deadlock if some filled channels are blocked, i.e., they contain a message that can never be forwarded by the target of the channel.
We show that every global deadlock is a weak deadlock, and every weak deadlock is a local deadlock. Furthermore, not every local deadlock is a weak deadlock. However, from a weak deadlock a local deadlock in the same network can be constructed.

Finally, we show how a packet switching network and the deadlock properties can be formalized using nuXmv~\cite{CCD+2014} and CTL~\cite{CE1982}. Our experiments indicate that different types of deadlock are found effectively in packet switching networks. However, verification times out due to the state space explosion when numbers of nodes and channels increase.

\paragraph{Structure of the paper}
In Section~\ref{section:preliminaries} we define packet switching networks and their semantics. Subsequently, in Section~\ref{section:deadlocks} we introduce three different notions of deadlock. Section~\ref{section:expressivity} makes a detailed comparison between these different notions.
In Section~\ref{section:implementation} we describe a translation of packet switching networks and deadlocks into nuXmv and CTL, and describe an experiment with this setup. Conclusions are presented in Section~\ref{section:conclusions}. This paper includes the full proofs of the presented results.

\section{Preliminaries}\label{section:preliminaries}

\subsection{Packet switching network}\label{subsection:packet-switching-networks}

A packet switching network consists of a set of nodes connected by (unidirectional) channels. A subset of the nodes is considered to be terminal.
Any node in the network can receive a message from an incoming channel and forward it to an outgoing channel. Terminal nodes can, furthermore, send messages into the network and receive messages from the network.
When forwarding a message or sending a message, this is always done in accordance with the routing function.
In this paper we consider networks with a static, deterministic routing function. The framework we present could be generalized to a non-deterministic setting. 
Formally, a packet switching network is defined as follows~\cite{DT2004}.

\begin{definition} A \textit{packet switching network} is a tuple $\N = (N,M,C,\route)$ where:
\begin{itemize}
    \item $N$ is a finite set of \textit{nodes},
    \item $M \subseteq N$ is the set of terminals, nodes that are able to send and receive messages, with $|M| \geq 2$,
    \item $C \subseteq N \times N$ is a finite set of \textit{channels}, and
    \item $\route \colon N \times M \to C$ is a (deterministic) routing function.
\end{itemize}
For channel $(n,m) \in C$ we write $\source((n,m)) = n$ and $\target((n,m)) = m$. 
We require $\source(\route(n,m)) = n$ for every $n \in N$, $m \in M$, with $m \neq n$. We write $c = m$ to denote that channel $c$ contains a message with destination $m$, and write $c = \bot$ to denote channel $c$ is empty. We write $M_\bot$ to denote $M \cup \{ \bot \}$.
\end{definition}
Routing function $\route$ decides the outgoing channel of node $n$ to which messages with destination $m$ should be forwarded.
For $m \in M$, $n \in N$ (with $m \neq n$), the next hop $\nexthop{m} \colon N \to N$ is defined as $\nexthop{m}(n) = \target(\route(n,m))$.


A packet switching network is \emph{correct} if, whenever a message is in a channel, the routing function is such that the message can reach its destination in a bounded number of steps. In essence, this means the routing function does not cause any messages to cycle in the network.
\begin{definition}\label{def:correctness}
Let $\N = (N, M, C, \route)$ be a packet switching network. The network is \emph{correct} if for every $m \in M$, and $n \in N$ there exists $k \geq 0$ such that
\[
\nexthop{m}^k(n) = m,
\]
where $\nexthop{n}^0(n) = n$ and $\nexthop{m}^{k+1}(n) = \nexthop{m}^k(\nexthop{m}(n))$.
\end{definition}
In our examples, we typically choose the routing function such that $k$ is minimal,
i.e., the routing function always follows the shortest path to the destination.
In any given state of the network a channel may be free, or it may be occupied by a message.
In the latter case it blocks access to that channel for other messages.
Processing in the network is \emph{asynchronous}, which means that at any moment a step can be done without central control by a clock. The content of a channel is identified by the destination $m \in M$ of the corresponding message.
More precisely, the following steps can be performed in a packet switching network:
\begin{description}
    \item[\emph{Send}] Terminal $m \in M$ can send a message to terminal $m' \in M$ by inserting a message in channel $\route(m,m')$, provided this channel is currently empty. After sending, this channel is occupied by $m'$.
    \item[\emph{Receive}] If channel $c \in C$ with $\target(c) = m$ contains a message with destination $m$, the message can be received by terminal $m$ and $c$ becomes free.
    \item[\emph{Process}] If channel $c \in C$ contains a message with destination $m \in M$, and $\target(c) = n \neq m$ for $n \in N$, then the message can be processed by node $n$ by forwarding it to channel $c' = \route(n,m)$.
    This step can only be taken if channel $c'$ is free. As a result, the message is removed from channel $c$ (which now becomes free) and moved to channel $c'$.
\end{description}
We illustrate these steps in a packet switching network in Example~\ref{exampleSteps}.

%
%
\begin{figure}[t]
\centering
\begin{subfigure}{.3\textwidth}
  \centering
  \begin{tikzpicture}[scale=2,align=center]
    \foreach \pos/\name in {{(0,1)/1}, {(1,1)/2}, {(1,0)/3},{(0,0)/4}}
        \node[vertex] (\name) at \pos {$\name$};
    \foreach \source/ \dest /\weight in {1/2/c_1=\bot}
        \path[arrow] (\source) -- node[weight,above] {$\weight$} (\dest);
    \foreach \source/ \dest /\weight in {3/4/c_3=\bot}
        \path[arrow] (\source) -- node[weight,below] {$\weight$} (\dest);
    \foreach \source/ \dest /\weight in {2/3/c_2=\bot, 4/1/c_4=\bot}
        \path[arrow] (\source) -- node[weight,left,midway] {$\weight$} (\dest);
    \foreach \vertex in {1,2,3,4}
        \path[vertex] node[selected vertex] at (\vertex) {$\vertex$};
        \end{tikzpicture}
        \caption{Initial state}
        \label{fig:example-steps-psn-init}
\end{subfigure}%
\begin{subfigure}{.3\textwidth}
  \centering
  \begin{tikzpicture}[scale=2,align=center]
   \foreach \pos/\name in {{(0,1)/1}, {(1,1)/2}, {(1,0)/3},{(0,0)/4}}
        \node[vertex] (\name) at \pos {$\name$};
    \foreach \source/ \dest /\weight in {1/2/c_1=3}
        \path[arrow] (\source) -- node[weight,above] {$\weight$} (\dest);
    \foreach \source/ \dest /\weight in {3/4/c_3=\bot}
        \path[arrow] (\source) -- node[weight,below] {$\weight$} (\dest);
    \foreach \source/ \dest /\weight in {2/3/c_2=\bot, 4/1/c_4=\bot}
        \path[arrow] (\source) -- node[weight,left,midway] {$\weight$} (\dest);
    \foreach \vertex in {1,2,3,4}
        \path[vertex] node[selected vertex] at (\vertex) {$\vertex$};
        \end{tikzpicture}
        \caption{send step}
        \label{fig:example-steps-psn-send}
\end{subfigure}%
\begin{subfigure}{.3\textwidth}
  \centering
  \begin{tikzpicture}[scale=2,align=center]
   \foreach \pos/\name in {{(0,1)/1}, {(1,1)/2}, {(1,0)/3},{(0,0)/4}}
        \node[vertex] (\name) at \pos {$\name$};
    \foreach \source/ \dest /\weight in {1/2/c_1=\bot}
        \path[arrow] (\source) -- node[weight,above] {$\weight$} (\dest);
    \foreach \source/ \dest /\weight in {3/4/c_3=\bot}
        \path[arrow] (\source) -- node[weight,below] {$\weight$} (\dest);
    \foreach \source/ \dest /\weight in {2/3/c_2=3, 4/1/c_4=\bot}
        \path[arrow] (\source) -- node[weight,left,midway] {$\weight$} (\dest);
    \foreach \vertex in {1,2,3,4}
        \path[vertex] node[selected vertex] at (\vertex) {$\vertex$};
        \end{tikzpicture}
        \caption{process step}
        \label{fig:example-steps-psn-process}
  \end{subfigure}%
    \caption{A packet switching network with send, process and receive steps}
    \label{fig:example-steps-psn}
\end{figure}
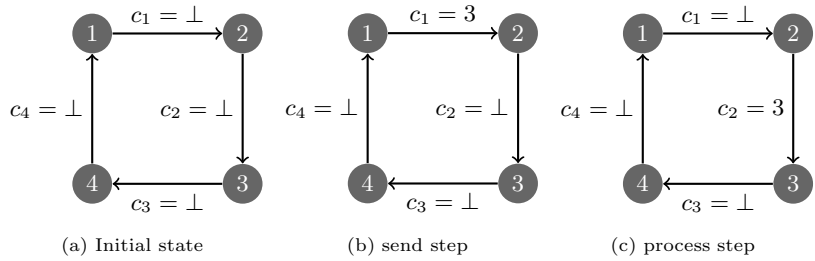
\begin{example}
\label{exampleSteps}
Consider the packet switching network in Figure~\ref{fig:example-steps-psn}.
The network consists of four nodes, i.e., $N = \{1,2,3,4\}$, all of which are terminals, so $M=N$, and four channels, $C= \{c_1,c_2,c_3,c_4\}$, shown as arrows from source to target.
The routing function is $\route(n,m) = c_n$ for all $n \in N$ and $m \in M$.
Initially, all channels are empty, this is shown in Figure~\ref{fig:example-steps-psn-init}. From the initial state it is possible to perform a \emph{send} step from any of the nodes. For example, since channel $c_1 = \bot$, a message can be sent from node $1$ to node $3$. The message is routed to $c_1$. The resulting state is shown in Figure~\ref{fig:example-steps-psn-send}.
Now, $c_1 = 3$ and $c_2 = \bot$, hence node $2$ can perform a \emph{process} step, and forward the message to $c_2$. The resulting situation is shown in Figure~\ref{fig:example-steps-psn-process}.
Finally, since $c_2 = 3$, and $\target(c_2) = 3$, node $3$ can execute a \emph{receive} step, and consume the message from channel $c_2$. Consequently, all channels are empty and the system is back to the initial state shown in Figure~\ref{fig:example-steps-psn-init}.
\end{example}

\subsection{Semantics of packet switching networks}\label{subsection:semantics-of-packet-switching-networks}
We formalize the semantics of packet switching networks using Kripke structures.
\begin{definition}\label{def:kripke-structure}
Let $\AP$ be a set of atomic propositions. A Kripke structure over $\AP$ is a four-tuple $K = (S, I, \to, L)$, where:
\begin{itemize}
  \item $S$ is a (finite) set of states,
  \item $I \subseteq S$ is the set of initial states,
  \item $\to \subseteq S \times S$ is the transition relation, which is total, i.e., for all $s \in S$ there exists $t \in S$ such that $s \to t$, and
  \item $L \colon S \to 2^{\AP}$ is a labelling function that assigns a set of atomic propositions to each state.
\end{itemize}
\end{definition}
In general, the set of states in a Kripke structure may be an overapproximation of the states that can be reached from an initial state.
In this paper we sometimes only consider the \emph{reachable} states of the system.

\begin{definition}\label{def:reachable}
Let $K = (S, I, \to, L)$ be a Kripke structure.
The set of \emph{reachable states} of $K$ is defined as follows:
\[
R(K) = \{ s' \in S \mid \exists s \in I \colon s \to^* s' \}
\]
where $\to^*$ denotes the reflexive transitive closure of $\to$.
\end{definition}




We now formalize the semantics of a packet switching network. This captures the intuitions described in Section~\ref{subsection:packet-switching-networks}.

\begin{definition}\label{def:semantics-of-packet-switching-networks}
Given packet switching network $\N = (N, M, C, \route)$, its semantics is defined as the Kripke structure $\KSN = (S, I, \to, L)$ over $\AP = \{ c = m \mid c \in C \land m \in M_\bot \}$, defined as follows:
\begin{itemize}
    \item $S = M_{\bot}^{|C|}$, i.e., the state of the network is the content of its channels. If $C = \{ c_1, \ldots, c_{|C|} \}$ we write $\pi_{c_i}(s) = v_i$ if $s = (v_1,\ldots,v_{|C|}) \in S$,
    \item $I = \{ s \in S \mid \forall c \in C \colon \pi_c(s) = \bot \}$, i.e., initially all channels are empty,
    \item transition relation $\to \subseteq S \times S$ is $\to_s \cup \to_p \cup \to_r$, where
    \begin{itemize}
    \item $\to_s$ is the least relation satisfying
    \[
    \frac
      {m,m' \in M \quad m \neq m' \quad c = \route(m,m') \quad v_{c} = \bot}
      {(v_1,\ldots,v_{c},\ldots, v_{|C|}) \to_s (v_1,\ldots,m',\ldots,v_{|C|})}
    \]
    characterising that terminal $m$ sends a message to terminal $m'$,
    \item $\to_p$ is the least relation satisfying
    \[
    \frac
      {m \in M \quad v_c = m \quad \target(c) = n \quad \route(n,m) = c' \quad v_{c'} = \bot
      }
      {(v_1,\ldots, v_{c}, \ldots, v_{c'}, \ldots, v_{|C|}) \to_p (v_1,\ldots, \bot, \ldots, m, \ldots, v_{|C|})}
    \]
    characterising that node $n$ forwards a message with destination $m$ that comes in on channel $c$ to channel $c'$, and
    \item $\to_r$ is the least relation satisfying
     \[
    \frac
      {m \in M \quad v_c = m \quad \target(c) = m}
      {(v_1,\ldots,v_{c},\ldots, v_{|C|}) \to_r (v_1,\ldots,\bot,\ldots,v_{|C|})}
    \]
    characterising that terminal $m$ receives a message along its incoming channel $c$.
    \end{itemize}
    
    \item $L(s)= \bigcup_{c \in C}\{c=m \mid \pi_{c}(s) = m \}$, for every $s \in S$.
\end{itemize}
Note that it is straightforward to show that $\to_s$, $\to_p$ and $\to_r$ are pairwise disjoint. We sometimes write, e.g., $\to_{pr}$ instead of $\to_p \cup \to_r$.
We write $\not \to_X$ if there is no $s' \in S$ such that $s \to_X s'$ for $x \subseteq \{s,p,r\}$.
To ensure that the transition relation is total, we extend $\to$ with transitions $s \to s$  whenever $s \not \to_{spr}$.
\end{definition}

\section{Deadlocks}\label{section:deadlocks}

The key question about packet switching we are interested in is whether a network is deadlock free.
Intuitively, a network contains a deadlock if a message is stuck in a channel, and it will never be processed or received by the target of the channel.
In practice, we can distinguish different notions of deadlock, each of which has a different interpretation of this informal requirement.
We introduce three such notions, and study the relation between them.

\subsection{Global deadlock}\label{subsection:global-deadlock}

Typically a global deadlock is a state that has no outgoing transitions.
However, since we are dealing with Kripke structures, which have a total transition relation, every state has an outgoing transition. A \emph{global deadlock} is, therefore, a state that has no outgoing transitions to a state other than itself.

\begin{definition}\label{def:globaldeadlock}
Let $K = (S, I, \to, L)$ be a Kripke structure.
The set of global deadlock states in $K$ is defined as:
\[
G(K) = \{ s \in S \mid \nexists s' \in S \colon s \neq s' \land s \to s' \}
\]
When $s \in G(K)$, we say that $s$ is a global deadlock.
\end{definition}

\begin{example}
\label{example:global-deadlock}
Recall the packet-switching network from Example~\ref{exampleSteps}. The situation in which all nodes have sent a message two hops away is shown on the right.
\noindent
\begin{minipage}{.70\textwidth}
All channels contain a value $m \in M$, but none of them can make progress because the next hop is blocked by another message.
For instance, message $3$ in $c_1$ has to reach node $3$, but $\route(2,3)=c_2$ is blocked by message $4$. There is a cycle of blocked channels, where all of them are filled, hence the network is in a global deadlock.
\end{minipage}%
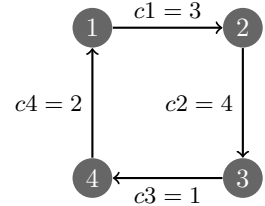
\hfill
\begin{minipage}{.29\textwidth}
\begin{tikzpicture}[scale=2]
    \foreach \pos/\name in {{(0,1)/1}, {(1,1)/2}, {(1,0)/3},{(0,0)/4}}
        \node[vertex] (\name) at \pos {$\name$};
    \foreach \source/ \dest /\weight in {1/2/c1=3}
        \path[arrow] (\source) -- node[weight,above] {$\weight$} (\dest);
    \foreach \source/ \dest /\weight in {3/4/c3=1}
        \path[arrow] (\source) -- node[weight,below] {$\weight$} (\dest);
    \foreach \source/ \dest /\weight in {2/3/c2=4, 4/1/c4=2}
        \path[arrow] (\source) -- node[weight,left] {$\weight$} (\dest);
    \foreach \vertex in {1,2,3,4}
        \path[vertex] node[selected vertex] at (\vertex) {$\vertex$};
\end{tikzpicture}
\end{minipage}
\end{example}

The semantics of packet switching networks guarantees that there are no global deadlocks among the initial states. 
\begin{lemma}\label{lemma:initial-global}
Let $\N = (N, M, C, \route)$ be a packet switching network with $\KSN = (S, I, \to, L)$ its semantics. Then $ I \cap G(\KSN) = \emptyset $
\end{lemma}
\begin{proof}
Since $|M| \geq 2$, and all channels are initially empty, there is a terminal node that can send a message into the network. \qedhere 
\end{proof}

\subsection{Local deadlock}\label{subsection:local-deadlock}

Even if not all of the channels in a packet switching network are blocked, it can happen that a subset of the channels is deadlocked. Such a situation is not covered by the global deadlock.
We therefore introduce the \emph{local deadlock}.
Intuitively, a state is a local deadlock if it has a channel that indefinitely contains the same message.

\begin{definition}\label{def:localdeadlock}
Let $\N = (N, M, C, \route)$ be a packet switching network, and $\KSN = (S, I, \to, L)$ its semantics.
The set of local deadlock states in $\KSN$ in which channel $c \in C$ is deadlocked is defined as:
\[
L_c(\KSN) = \{ s \in S \mid \forall s' \in S \colon s \to^* s' \implies \pi_c(s) \neq \bot \land \pi_c(s') = \pi_c(s) \}
\]
The set of local deadlock states is defined as:
\[
L(\KSN) = \bigcup_{c \in C} L_c(\KSN)
\]
\end{definition}

We illustrate the local deadlock in the following example.

\begin{example}
\label{example:local-deadlock}
Consider the packet switching network with $N=M=\{1,2,3,4\}$ and $C=\{c_1,c_2,c_3,c_4,c_5\}$ shown on the right. The routing function $\route(n,m) = c_5$ if $n = 3$ and $m = 2$, and $c_n$ otherwise. None of the messages in channels $c_1$, $c_2$, $c_3$ and $c_4$ can make another step because the next hop is blocked. For instance, message $4$ in $c_2$ has to reach node $4$, but $\route(3,4)=c_3$ is blocked by message~$1$. 

\noindent
\begin{minipage}{.69\textwidth}
Channel $c_5$, by definition of the routing function, is only used in case node $3$ sends a message to node $2$, $rout(3,2)=c5$. Therefore, node $3$ can still send such a message (which can be received by node $2$ immediately afterwards). Thus, these two steps will always be possible, even if all of the other channels are deadlocked.
\end{minipage}%
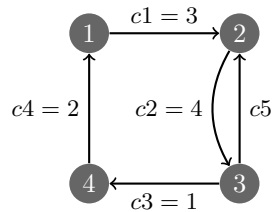
\hfill
\begin{minipage}{.3\textwidth}
\begin{tikzpicture}[scale=2]
    \foreach \pos/\name in {{(0,1)/1}, {(1,1)/2}, {(1,0)/3},{(0,0)/4}}
        \node[vertex] (\name) at \pos {$\name$};
    \foreach \source/ \dest /\weight in {1/2/c1=3}
        \path[arrow] (\source) -- node[weight,above] {$\weight$} (\dest);
    \foreach \source/ \dest /\weight in {3/4/c3=1}
        \path[arrow] (\source) -- node[weight,below] {$\weight$} (\dest);
    \foreach \source/ \dest /\weight in {4/1/c4=2}
        \path[arrow] (\source) -- node[weight,left] {$\weight$} (\dest);
    \foreach \source/ \dest /\weight in {2/3/c2=4}
        \path[arrow] (\source) edge[bend right] node[weight,left] {$\weight$} (\dest);
    \foreach \source/ \dest /\weight in {3/2/c5}
        \path[arrow] (\source) -- node[weight,right] {$\weight$} (\dest);
    \foreach \vertex in {1,2,3,4}
        \path[vertex] node[selected vertex] at (\vertex) {$\vertex$};
    \end{tikzpicture}
\end{minipage}
\end{example}

\subsection{Weak deadlock}\label{subsection:weak-deadlock}

Local deadlock does not distinguish between sending a new message---which is always possible if the target channel is empty---, and processing or receiving a message. In this section, we introduce the notion of \emph{weak deadlock}. A state is a weak deadlock if no \emph{receive} or \emph{process} step is possible in that state.

Before defining weak deadlock, we first observe that in the initial states of a Kripke structure representing a packet-switching network, trivially no \emph{process} or \emph{receive} step is possible.

\begin{lemma}\label{lem:initial-state-no-pr-transition}
Let $\N = (N, M, C, \route)$ be a packet switching network, and $\KSN = (S, I, \to, L)$ its semantics. Then $ \forall s \in I \colon s \centernot \to_{pr} $
\end{lemma}
\begin{proof}
Initially all channels are empty. The result then follows immediately from the definitions of $\to_r$ and $\to_p$. \qedhere
\end{proof}


Because of this observation, we explicitly exclude the initial states. The definition of weak deadlocks is as follows.
\begin{definition}\label{def:weakdeadlock}
Let $\N = (N, M, C, \route)$ be a packet switching network, and $\KSN = (S, I, \to, L)$  its semantics. The set of weak deadlocks is defined as:
\[
W(\KSN) = \{ s \in S \setminus I \mid s \not \to_{pr} \}
\]
\end{definition}
\begin{example}\label{example:weak-deadlock}
Consider the packet switching network with $N = M = \{1,2,3,4\}$ and $C=\{c_1,c_2,c_3,c_4,c_5\}$ shown on the right. The routing function $\route(n,m) = c_5$ if $n = 2$ and $m = 1$, and $c_n$ otherwise. None of the messages in $c_1, c_2, c_3, c_4$ can reach its destination because the next hop is blocked. For instance, message $2$ in $c_4$ has to reach node $2$, but $\route(1,2)=c_1$ is blocked by message $3$.

\smallskip
\noindent
\begin{minipage}{.72\textwidth}
Channel $c_5$, by definition of the routing function, is used only when node $2$ sends a message to node $1$, $rout(2,1)=c_5$.
This means $c_5$ can be filled with value $1$, after which it can be received immediately by node $1$. Thus, node $2$, through channel $c_5$, will always be able to send messages to node $1$, but in this particular configuration no \emph{process} or \emph{receive} step is possible. Hence, this situation is a weak deadlock.
\end{minipage}
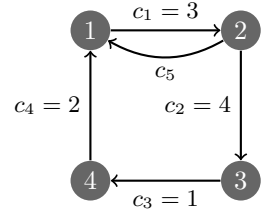
\hfill
\begin{minipage}{.28\textwidth}
\begin{tikzpicture}[scale=2,auto]
    \foreach \pos/\name in {{(0,1)/1}, {(1,1)/2}, {(1,0)/3},{(0,0)/4}}
        \node[vertex] (\name) at \pos {$\name$};
    \foreach \source/ \dest /\weight in {1/2/c_1=3}
        \path[arrow] (\source) -- node[weight,above] {$\weight$} (\dest);
    \foreach \source/ \dest /\weight in {3/4/c_3=1}
        \path[arrow] (\source) -- node[weight,below] {$\weight$} (\dest);
    \foreach \source/ \dest /\weight in {2/1/c_5}
        \path[arrow] (\source) edge[bend left] node[weight,below] {$\weight$} (\dest);
    \foreach \source/ \dest /\weight in {2/3/c_2=4, 4/1/c_4=2}
        \path[arrow] (\source) -- node[weight,left] {$\weight$} (\dest);
    \foreach \vertex in {1,2,3,4}
        \path[vertex] node[selected vertex] at (\vertex) {$\vertex$};
    \end{tikzpicture}
\end{minipage}
\end{example}




\section{Expressivity of different notions of deadlock}\label{section:expressivity}

In this section we compare the different notions of deadlock introduced in the previous section. We first relate global deadlocks to local and weak deadlocks, and ultimately we investigate the relation between local and weak deadlocks.

\subsection{Comparing global deadlocks to local and weak deadlocks}

It is not hard to see that every global deadlock is both a local deadlock and a weak deadlock. Furthermore, neither local nor weak deadlocks necessarily constitute a global deadlock.

We first formalize this for local deadlocks in the following lemma.
\begin{lemma}\label{lem:global-vs-local-deadlock}
Let $\N = (N, M, C, \route)$ be a packet switching network, and $\KSN = (S, I, \to, L)$ its semantics. Then we have:
\[
  G(\KSN) \subseteq L(\KSN)
 \]
\end{lemma}
\begin{proof}
From the Definitions~\ref{def:globaldeadlock} and~\ref{def:localdeadlock} it follows immediately that for all $c \in C$, $G(\KSN) \subseteq L_c(\KSN)$, hence $G(\KSN) \subseteq \bigcup_{c \in C} L_c(\KSN) = L(\KSN)$. \qedhere
\end{proof}

It is not generally the case that $L(\KSN) \subseteq G(\KSN)$. This follows immediately from Example~\ref{example:local-deadlock}, which shows a packet-switching network with a local deadlock that is not a global deadlock.
For weak deadlocks, similar results hold as formalized by the following lemma.

\begin{lemma}\label{lem:global-vs-weak-deadlock}
Let $\N = (N, M, C, \route)$ be a packet switching network, and $\KSN = (S, I, \to, L)$ its semantics. Then we have
\[
G(\KSN) \subseteq W(\KSN)
\]
\end{lemma}
\begin{proof}
Fix $s \in G(\KSN)$. Note that $\nexists s' \in S \colon s \neq s' \land s \to s'$ according to Definition~\ref{def:globaldeadlock}.
Towards a contradiction, suppose $s \to_{pr} s'$ for some $s'$. It follows from the definition of $\to_{pr}$ that $s \neq s'$, and since $\to_{pr} \subseteq \to$, this is a contradiction. So, $s \not \to_{pr}$. Hence according to Definition~\ref{def:weakdeadlock}, $s \in W(\KSN)$. So, $G(\KSN) \subseteq W(\KSN)$. \qedhere
\end{proof}
Again, the converse does not necessarily hold.
This follows immediately from Example~\ref{example:weak-deadlock}, which shows a packet switching network with a weak deadlock that is not a global deadlock.

\subsection{Comparing local deadlocks to weak deadlocks}

Now that we have shown that local and weak deadlocks are not necessarily global deadlocks, the obvious question is how local and weak deadlocks are related. In particular, what we show in this section is that there is a local deadlock in a packet switching network if, and only if, there is a weak deadlock in the network.

Before we prove this main result, we first present several lemmata supporting the proof.
%
%
%
%
First, in subsequent results we have to reason about the number of \emph{process} and \emph{receive} transitions that can be taken from a particular state, provided that no \emph{send} transitions are taken.
To this end, we first formalize the number of steps required to reach the destination for message $m$ in channel $c$.

\begin{definition}\label{def:N}
Let $c \in C$ be a channel, and $m \in M_{\bot}$ the destination of the message carried by the channel.
\[
N(c,m) = \begin{cases}
    0 & \text{if}\ m = \bot \\
    1 & \text{if}\ m = \target(c) \\
    1 + N(\route(\target(c),m), m) & otherwise
\end{cases}
\]
\end{definition}


To determine that $N$ is well-defined for correct packet switching networks, we first introduce $\nextC{m}$, that, in a similar way to $\nexthop{m}$ counts the number of channels that needs to be traversed for message $m$ in channel $c$ to reach its destination.

\begin{definition}\label{def:nextC}
Let $\N = (N, M, C, route)$ be a packet switching network, with 
$c \in C$ and $m \in M$ such that channel $c$ contains message $m$.
The next channel for $m$ in $c$ is defined as follows:
\[
\nextC{m}(c) = \route(\target(c),m).
\]
\end{definition}
The relation between $\nextC{m}$ and $\nexthop{m}$, is formalized in the following lemma.

\begin{lemma}\label{lem:next-vs-nextC}
Let $\N = (N, M, C, route)$ be a packet switching network.
For all $m \in M$, $k \geq 0$, $n \in N$ and $c \in C$ such that $\target(c) = n$, we have
\[
\nexthop{m}^k(n) = \target(\nextC{m}^k(c))
\]
\end{lemma}
\begin{proof}
Fix arbitrary $m \in M$.
We proceed by induction on $k$.
\begin{itemize}
    \item $k = 0$. Fix $n \in N$ and $c \in C$ such that $\target(c) = n$. Then $\nexthop{m}^{0}(n) = n$ by definition of $\nexthop{m}$. Since $\target(c) = n$, and by definition of $\nextC{m}^0(c) = c$, we find that $n = \target(\nextC{m}^0(c))$.
    
    \item $k = l + 1$. Fix $n \in N$ and $c \in C$ such that $\target(c) = n$. As induction hypothesis, assume that for all $n \in N$ and $c \in C$ such that $\target(c) = n$, 
    $\nexthop{m}^l(n) = \target(\nextC{m}^l(c))$.
    We derive as follows:
    \begin{align*}
    \nexthop{m}^{l+1}(n)
    & = \nexthop{m}^{l}(\nexthop{m}(n))
    && \text{Definition of $\nexthop{m}^{l}$}
    \\
    & = \nexthop{m}^{l}(\target(\route(n,m)))
    && \text{Definition of $\nexthop{m}$}
    \\
    & = \target(\nextC{m}^l(\route(n,m)))
    && \text{Induction hypothesis}
    \\
    & = \target(\nextC{m}^l(\route(\target(c),m)))
    && \text{$n = \target(c)$}
    \\
    & = \target(\nextC{m}^l(\nextC{m}(c)))
    && \text{Definition of $\nextC{m}$}
    \\
    & = \target(\nextC{m}^{l+1}(c))
    && \text{Definition of $\nextC{m}^{l+1}$}
    \end{align*}
    \qedhere
\end{itemize}
\end{proof}
In essence, $N(c,m)$ is an inductive characterization of $\nextC{m}(c)$. This correspondence is formalized as follows.

\begin{lemma}\label{lem:nextC-vs-N}
Let $\N = (N,M,C, \route)$ be a correct packet switching network.
For all $l \in \mathbb{N}$, channels $c \in C$ and messages $m \in M_\bot$, if $l$ is the smallest value such that there exists $c' \in C$ with $\nextC{m}^l(c) = c'$ and $\target(c') = m$, then $N(c,m) = l+1$.
\end{lemma}
\begin{proof}
We proceed by induction on $l$.
\begin{itemize}
    \item $l = 0$. Fix $c \in C$ and $m \in M_\bot$ such that $\nextC{m}^0(c) = c'$ with $\target(c') = m$. By definition of $\nextC{m}^0$, $c = c'$, hence $\target(c) = m$, and by definition of $N$, $N(c,m) = 1 = 0+1$.
    \item $l = k+1$. Fix $c \in C$ and $m \in M_\bot$ such that $k+1$ is the smallest value such that there exists $c' \in C$ with $\nextC{m}^{k+1}(c) = c'$ and $\target(c') = m$. Let $c'$ be such.
    By definition, $\nextC{m}^{k+1}(c) = \nextC{m}^{k}(\nextC{m}(c)) = \nextC{m}^{k}(\route(\target(c),m)) = c'$.
    Note that $k$ is the smallest value such that $\nextC{m}^{k}(\route(\target(c),m)) = c'$, otherwise this would contradict that $k+1$ is the smallest such value for channel $c$. Therefore, according to the induction hypothesis, $N(\route(\target(c),m)) = k + 1$, since $m \neq \bot$ and $m \neq \target(c)$, $N(c,m) = 1 + N(\route(\target(c),m),m) = 1 + k + 1 = l + 1$. \qedhere
\end{itemize}
\end{proof}

\begin{lemma}\label{lemma:N}
Let $\N = (N,M,C,rout)$ be a correct packet switching network, then for all channels $c \in C$ and messages $m \in M_\bot$, there exists $l \in \mathbb{N}$ such that $N(c,m) = l$.
\end{lemma}

\begin{proof}
Fix $c \in C$ and $m \in M_\bot$. Note that if $m = \bot$, then $N(c,m) = 0$, so the result follows immediately.
Now, assume that $m \neq \bot$. Since the network is correct, there must be some $k \in \mathbb{N}$ such that $\nexthop{m}^k(\target(c)) = m$. Pick the smallest such $k$.
According to Lemma~\ref{lem:next-vs-nextC}, $\nexthop{m}^k(\target(c)) = \nextC{m}^k(c)$. 
Hence there exists channel $c'$ such that $\nextC{m}^k(c) = c'$ with $\target(c') = m$ by definition of $\nextC{m}^k(c)$. Also, there is no $l < k$ such that $\target(\nextC{m}^l(c)) = m$,  since otherwise we would have a contradiction.
Therefore, according to Lemma~\ref{lem:nextC-vs-N}, $N(c,m) = k+1$.\qedhere
\end{proof}
We use this property to show that, from a given state in a packet switching network, if we only execute \emph{process} or \emph{receive} steps, the number of steps that can be taken is finite.

\begin{lemma}\label{lem:finite-process-or-receive-steps}
Let $\N = (N,M,C,rout)$ be a correct packet switching network with $\KSN=(S,I,\to,L)$ its semantics, then
\[
\forall s \in S \colon \exists s' \in S \colon s \to_{pr}^* s' \wedge s' \not \to_{pr}
\]
i.e., the number of possible steps of type $\to_{pr}$, from state $s$, is bounded.
\end{lemma}

\begin{proof}
We prove that the number of possible steps of type $\to_{pr}$, starting from state $s \in S$, is finite. This means that eventually it will not be possible to do a $\to_{pr}$ step anymore.

Using $N(c,m)$ from Definition~\ref{def:N} we define the weight of a state $s \in S$ as follows:
\begin{equation*}
    wt(s) = \sum_{c \in C} N(c, \pi_{c}(s))
\end{equation*}
The weight captures the total number of steps required such that all the messages currently in the network can reach their destination. Note that $N$ is well-defined according to Lemma~\ref{lemma:N}, hence $wt$ is well-defined.

We now prove for all $s, s' \in S$ that if $s \to_{pr} s'$, then $wt(s') < wt(s)$.
Fix $s, s' \in S$ such that $s \to_{pr} s'$. Note that $wt(s) = \sum_{c \in C} N(c, \pi_{c}(s))$.
We distinguish two cases:
\begin{itemize}
    \item $s \to_{p} s'$. Then there must be $m \in M$ and $c, c' \in C$ such that $\pi_c(s) = m$, $\route(\target(c),m) = c'$ and $\pi_{c'}(s) = \bot$, and $\pi_{c}(s') = \bot$, $\pi_{c'}(s') = m$, and for all $c'' \in C \setminus \{ c,c' \}$, $\pi_{c''}(s) = \pi_{c''}(s')$.
    Let $m$, $c$ and $c'$ be such.
    
    Note that by definition of $N$, $N(c,m) = 1 + N(\route(\target(c),m),m) = 1 + N(c',m)$. Therefore, $N(c,\pi_c(s)) + N(c', \pi_{c'}(s)) = 1 + N(c,\pi_c(s')) + N(c', \pi_{c'}(s'))$; also $N(c'', \pi_{c''}(s)) = \pi_{c''}(s')$ for all $c \in C \setminus \{ c', c'' \}$, hence $wt(s) = \sum_{c \in C}N(c, \pi_c(s)) > \sum_{c \in C}N(c, \pi_c(s')) = wt(s')$.
    
    \item $s \to_{r} s'$. Then there must be $m \in M$, $c \in C$ such that $\pi_c(s) = m$ and $\target(c) = m$ and $\pi_{c}(s') = \bot$, and for all $c' \in C \setminus \{ c \}$, $\pi_{c'}(s) = \pi_{c'}(s')$. Let $m$ and $c$ be such.
    
    Note that $N(c, \pi_{c}(s)) = 1$, $N(c, \pi_{c}(s')) = 0$, and for all $c' \in C \setminus \{c\}$, $N(c', \pi_{c'}(s)) = N(c', \pi_{c'}(s'))$. Hence, we have $wt(s) = \sum_{c \in C}N(c, \pi_c(s)) > \sum_{c \in C} N(c, \pi_c(s')) = wt(s')$.
\end{itemize}

So, the weight of the state decreases on every transition taken in $\to_{pr}$. Note that it follows immediately from the definition of $N$ that, if there is a $\to_{pr}$ transition from state $s$, then for some channel $c$ and message $\pi_c(s)$, $N(c,\pi_c(s)) > 0$.
Therefore, the number of $\to_{pr}$ steps is finite.
Hence, for all states $s$ in $K$, there is a state $s'$ such that $s \to_{pr}^* s'$ such that $s' \not \to_{pr}$. \qedhere
\end{proof}

At this point, we can finally formalize the correspondence between weak and local deadlocks.
We first prove that a weak deadlock is also a local deadlock.

\begin{theorem}\label{prop:weak-is-local-deadlock}
Let $\N = (N,M,C,rout)$ be a correct packet switching network and $\KSN=(S,I,\to,L)$ its semantics. Then we have
\[
W(\KSN) \subseteq L(\KSN)
\]
\end{theorem}
\begin{proof}
Fix arbitrary $s \in W(\KSN)$. From the definition of $W(\KSN)$, we observe that $s \not \in I$ and $s' \not \to_{pr}$.
Let $C' = \{ c \in C \mid \pi_c(s) \neq \bot \}$ be the set of non-empty channels in state $s$.
Since $s \not \in I$, $C' \neq \emptyset$.

Observe that for all $c \in C'$, $\pi_c(s) \neq \target(c)$, and $\route(\target(c),\pi_c(s)) \in C'$ from the definitions of $\to_p$ and $\to_r$, since $s \not \to_{pr}$.

Next we show that for all $s' \in S$ such that $s \to^* s'$, for all $c \in C'$, $\pi_c(s') = \pi_c(s)$. We proceed by induction.
If $s \to^0 s'$, then $s = s'$ and the result follows immediately.
Now, assume there exists $s''$ such that $s \to^n s'' \to s'$.
According to the induction hypothesis, for all $c \in C'$, $\pi_{c}(s'') = \pi_{c}(s)$. Fix arbitrary $c \in C'$.
It follows from our observations that $\pi_c(s'') \neq \target(c)$ and $\route(\target(c),\pi_c(s'')) \in C'$, hence $\route(\target(c),\pi_c(s'')) \neq \bot$.
Therefore, the only possible transitions are a self-loop in which $s'' \to s'$ with $s'' = s'$, or a transition $\to_r$, in which case $\pi_c(s') = \pi_c(s'')$ according to the definition of $\to_r$.

Hence, it follows that $s \in L_c(\KSN)$ for all $c \in C'$, and since $c'$ is non-empty, $s \in L(\KSN)$. So $W(\KSN) \subseteq L(\KSN)$. \qedhere
\end{proof}
The following example shows that generally not $L(\KSN) \subseteq W(\KSN)$.

\begin{example}\label{example:reachable-local-is-not-weak}
Consider the packet switching network with $N = M = \{ 1,2,3,4 \}$ and $C = \{ c_1, c_2, c_3, c_4, c_5 \}$ shown on the right.\footnote{This is the same network as in Example~\ref{example:weak-deadlock}, but with $c_5 = 1$ instead of $c_5 = \bot$.} Note that this configuration is reachable by applying the following \emph{send} steps in any order:

\noindent
\begin{minipage}{.70\textwidth}
from node $1$ to node $3$, from node $2$ to node $4$, from node $2$ to node $1$, from node $3$ to node $1$, and from node $4$ to node $2$. This results in the configuration we show.
This is a local deadlock for channels $c_1$ through $c_4$. Note, however, that node $1$ can receive the message from channel $c_5$, so this is not a weak deadlock.
\end{minipage}
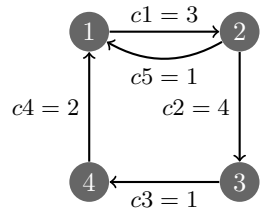
\hfill
\begin{minipage}{.31\textwidth}
\begin{tikzpicture}[scale=2,auto]
    \foreach \pos/\name in {{(0,1)/1}, {(1,1)/2}, {(1,0)/3},{(0,0)/4}}
        \node[vertex] (\name) at \pos {$\name$};
    \foreach \source/ \dest /\weight in {1/2/c1=3}
        \path[arrow] (\source) -- node[weight,above] {$\weight$} (\dest);
    \foreach \source/ \dest /\weight in {3/4/c3=1}
        \path[arrow] (\source) -- node[weight,below] {$\weight$} (\dest);
    \foreach \source/ \dest /\weight in {2/1/c5=1}
        \path[arrow] (\source) edge[bend left] node[weight,below] {$\weight$} (\dest);
    \foreach \source/ \dest /\weight in {2/3/c2=4, 4/1/c4=2}
        \path[arrow] (\source) -- node[weight,left] {$\weight$} (\dest);
    \foreach \vertex in {1,2,3,4}
        \path[vertex] node[selected vertex] at (\vertex) {$\vertex$};
    \end{tikzpicture}
\end{minipage}
\end{example}

The essence of a local deadlock is a cycle of nodes, each of which is waiting for an outgoing channel to become free. This suggests that from a local deadlock, we can construct a weak deadlock by removing all messages that do not play a role in such a cycle. This is what we prove in the following theorem.

\begin{theorem}\label{prop:reachable-local-is-weak-deadlock}
Let $\N = (N,M,C,rout)$ be a correct packet switching network and $\KSN=(S,I,\to,L)$ its semantics. Then we have
\[
L(\KSN) \neq \emptyset \implies W(\KSN) \neq \emptyset
\]
\end{theorem}
\begin{proof}
Fix arbitrary $s \in L(\KSN)$. We show that from $s$ we can reach a state $s'$ such that $s' \in W(\KSN)$.
Since $s \in L(\KSN)$, there exists a channel $c \in C$ such that $s \in L_c(\KSN)$, hence  $\pi_c(s) \neq \bot$ and for all $s' \in S$ such that $s \to^* s'$, we have $\pi_c(s') = \pi_c(s)$. According to Lemma~\ref{lem:finite-process-or-receive-steps}, there exists $s'$ such that $s \to_{pr}^* s'$ and $s' \not \to_{pr}$. Pick such $s'$. Since $s \to^* s'$, we have $\pi_c(s') = \pi_c(s)$ and $\pi_c(s') \neq \bot$. Note that since $\pi_c(s') \neq \bot$, $s' \not \in I$. Hence $s' \in W(\KSN)$, so $W(\KSN) \neq \emptyset$. \qedhere
\end{proof}

\section{Proof of Concept Implementation}
\label{section:implementation}
In this section we present a proof-of-concept implementation of the theory formalized in this paper. We translate packet switching networks into SMV, translating the different notions of deadlock to CTL. We use nuXmv~\cite{CCD+2014,CCD+2016} to find deadlocks in the models.
We describe examples to which these tools have been applied, and evaluate the results.
In the rest of this section, fix packet switching network  $\N = (N,M,C,\route)$, with Kripke structure $K_{\N}=(S,I,\to,L)$. For channels $c \in C$ and nodes $m \in M$, we use $\nextC{m}(c) = \route(\target(c),m)$ to denote the next channel for message $m$ when it is currently in $c$.

\subsection{An SMV model for packet switching networks}\label{subsection:model}

We sketch the translation of packet switching network $\N$ to the SMV format used by nuXmv.
The SMV model consists of the  following parts:
\begin{itemize}
    \item \textsc{declarations}: $c_i \colon 1 \ldots |N|$. That is, the model has a variable $c_i$ for every channel $c_i \in C$. The value of the channels is in the range $0 \ldots |N|$. Note that $c_i = 0$ encodes $c_i = \bot$, i.e., the empty channel.
     \item \textsc{initialization:} $\bigwedge_{i=1}^{|C|} c_i = 0$. That is, initially all channels are empty.
     \item \textsc{transition relation:} The transition relation is the disjunction over all send, process and receive transitions that are specified as follows.
     For each $c_i \in C$ such that $\source(c_i) \in M$ (i.e. its source is terminal), and message $m \neq \source(c_i)$ that $c_i$ can insert into the network, we have a \textsc{send} transition:
         \[
         \begin{array}{llll}
             \texttt{case} & c_i = 0 & \colon \texttt{next}(c_i) = m \land \bigwedge_{j \neq i} \texttt{next}(c_j) = c_j; & \\
                           & \texttt{TRUE} & \colon \bigwedge_{i=1}^{|C|} \texttt{next}(c_i) = c_i; & \texttt{esac}
         \end{array}
         \] 
         
    For all channels $c_i, c_j \in C$ and messages $m$, such that $\nextC{m}(c_i) = c_j$, we have the following \textsc{process} transition:
         \[
         \begin{array}{ll}
             \texttt{case} \\
             ~~~~c_i = m \land c_j = 0 & \colon \texttt{next}(c_i) = 0 \land \texttt{next}(c_j) = m \land  \bigwedge_{k \not \in \{ i, j \}} \texttt{next}(c_k) = c_k; \\
             ~~~~\texttt{TRUE} & \colon \bigwedge_{i=1}^{|C|} \texttt{next}(c_i) = c_i; \\ 
             \texttt{esac}
         \end{array}
         \]
         
    For all channels $c_i \in C$ and messages $m$ such that $\target(c_i) = m$, we have the following \textsc{receive} transition:
         \[
         \begin{array}{llll}
             \texttt{case} & c_i = m & \colon \texttt{next}(c_i) = 0 \land \bigwedge_{j \neq i} \texttt{next}(c_j) = c_j; & \\
                           & \texttt{TRUE} & \colon \bigwedge_{i=1}^{|C|} \texttt{next}(c_i) = c_i; & \texttt{esac}
         \end{array}
         \]
\end{itemize}
In this encoding \texttt{next} returns the value of its argument in the next state.

\subsection{Deadlock formulas in CTL}
To find deadlocks using nuXmv, we translate the properties to CTL.
For the sake of readability, we give the CTL formulas as defined for the Kripke structures. These formulas are easily translated into the explicit syntax of nuXmv.
%

\begin{definition}\label{def:ctl-global-deadlock}
The CTL formula for global deadlock is the following: 
\begin{multline*}
\mathsf{EF}(\neg(\bigvee_{c \in C} v_c=\bot \vee \bigvee_{c \in C}\bigvee_{m \in M}(\target(c) \neq m \wedge v_c = m \wedge v_{\nextC{m}(c)} = \bot) \vee 
\\ \bigvee_{c \in C}\bigvee_{m \in M}( \target(c)=m \wedge v_c = m))).
\end{multline*}
\end{definition}
This formula expresses that a state can be reached, in which non of the conditions required to take a transition holds. The disjuncts are the conditions for send, process and receive transitions, respectively.



\begin{definition}
Local deadlock is defined in CTL as follows:
\[
    \bigvee_{c \in C} \bigvee_{m \in M}(\mathsf{EF}(\mathsf{AG}(v_c = m))).
\]
\end{definition}
This formula checks whether, a state can be reached in which, for some channel $c$ and message $m$, $c$ contains $m$, and $m$ can never be removed from $c$.




\begin{definition}
Weak deadlock is defined in CTL as follows:
\begin{multline*}
\mathsf{EF}(\bigvee_{c \in C}(v_c \neq \bot) \land \neg(\bigvee_{c \in C}\bigvee_{m \in M}(\target(c) \neq m \wedge v_c = m \wedge v_{\nextC{m}(c)} = \bot) \vee \\ \bigvee_{c \in C}\bigvee_{m \in M}( \target(c)=m \wedge v_c = m))).
\end{multline*}
\end{definition}
This formula expresses a non-initial state can be reached in which no \emph{process} or \emph{receive} transition is enabled. 
The conditions are the same as in Definition~\ref{def:ctl-global-deadlock}.


\begin{figure}[b]
    \centering
    \begin{subfigure}{.5\textwidth}
    \centering
      \begin{tikzpicture}[scale=1,auto]
    \foreach \pos/\name in {{(0,4)/1}, {(3,0)/10}, {(2,0)/11}, {(1,0)/12}, {(0,0)/13}, {(1,4)/2}, {(2,4)/3}, {(3,4)/4}, {(4,4)/5}, {(4,3)/6}, {(4,2)/7}, {(4,1)}/8, {(4,0)}/9, {(0,1)/14}, {(0,2)/15}, {(0,3)/16}, {(2,2)/17}}
        \node[vertex] (\name) at \pos {$\name$};
    \foreach \source/ \dest in {1/2, 2/3, 3/4, 3/17, 4/5, 4/6, 5/6, 6/7, 7/8, 7/17, 8/9, 8/10, 9/10, 10/11, 11/12, 11/17, 12/13, 13/14, 14/15, 15/16, 15/17, 16/1, 16/2, 17/3, 17/7, 17/11, 17/15}
        \path[arrow] (\source) -- (\dest);
\end{tikzpicture}
\captionof{figure}{Network as directed graph}
\label{fig:17}
    \end{subfigure}%
    \begin{subfigure}{.5\textwidth}
    \centering
      \begin{tikzpicture}[scale=0.95,auto]
    \foreach \pos/\name in {{(0,4)/1}, {(3,0)/10}, {(2,0)/11}, {(1,0)/12}, {(0,0)/13}, {(1,4)/2}, {(2,4)/3}, {(3,4)/4}, {(4,4)/5}, {(4,3)/6}, {(4,2)/7}, {(4,1)}/8, {(4,0)}/9, {(0,1)/14}, {(0,2)/15}, {(0,3)/16}, {(2,2)/17}}
        \node[vertex] (\name) at \pos {$\name$};
    \foreach \source/ \dest/ \weight in {15/17/11, 17/11/12, 11/12/13, 12/13/11, 13/14/11, 14/15/11}
       \path[arrow] (\source) -- node[weight] {$\weight$} (\dest);
    \foreach \source/ \dest in {1/2, 2/3, 3/4, 3/17, 4/5, 4/6, 5/6, 6/7, 7/8, 7/17, 8/9, 8/10, 9/10, 10/11, 15/16, 16/1, 16/2, 17/3, 17/7}
        \path[arrow] (\source) -- (\dest);
    \foreach \source/ \dest in {17/15, 11/17}
        \path[arrow] (\source) edge[bend left] (\dest);
    \foreach \vertex in {1,5,8,11,12,13,15}
        \path[vertex] node[selected vertex] at (\vertex) {$\vertex$};
\end{tikzpicture}
\captionof{figure}{Local deadlock $M=\{1,5,8,11,12,13,15\}$}
\label{fig:xMASLDex}
    \end{subfigure}
    \caption{17 nodes and 27 channels packet switching network}
    \label{fig:17PSN}
\end{figure}
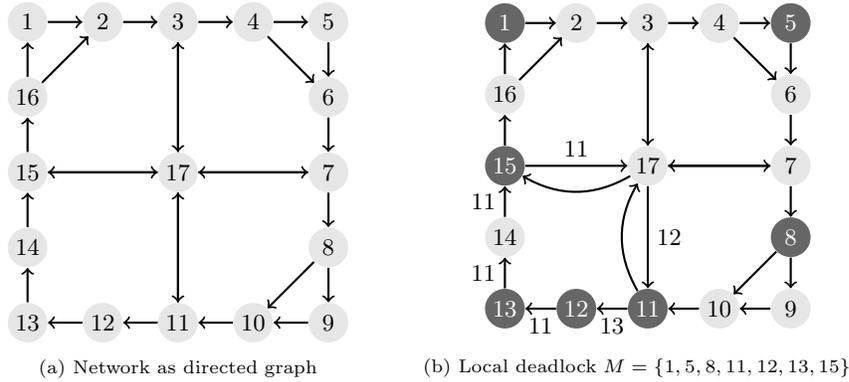

\subsection{Experiment}
We evaluate our proof-of-concept implementation on a packet switching network that consists of 17 nodes and 27 channels. The network is shown in Figure \ref{fig:17}. The nodes are numbered consecutively from 1 to 17. Note that double-ended arrows represent pairs of channels; one channel per direction. The routing function used is the shortest path, which is unique for every pair of nodes $n$ and $n'$. 
We vary the set of terminals in this network, and determine for each of the notions of deadlock whether a deadlock exists. A timeout is set at 2 hours and 30 minutes.

The experiments were done using nuXmv 1.1.1, on a system running Windows 10 Home, 64 bit Intel(R) Core(TM) i7-7500U CPU @ 2.70GHz and 8GB of RAM.

\begin{table}[!t]
\centering
\caption{Deadlock and execution times (s) using nuXmv for the network in Fig.~\ref{fig:17}}
\label{tab:results17}
\begin{tabular}{@{}l|cr|cr|cr@{}}
\toprule
\textbf{Set of terminals M} & \multicolumn{2}{c|}{\textbf{Global}} & \multicolumn{2}{c|}{\textbf{Local}}                             & \multicolumn{2}{c}{\textbf{Weak}}                                        \\ 
                            & \textbf{dl}      & \textbf{time}            & \textbf{dl}               & \textbf{time}                              & \textbf{dl}               & \textbf{time}                                        \\ \midrule
\{2,4,6\}                   & d                & 0.74                      & d                         & 0.71                                        & d                         & 0.71                                                  \\
\{1,8,10\}                  & d                & 0.95                      & d                         & 0.83                                        & d                         & 0.96                                                  \\
\{5,12,14\}                 & d                & 0.92                      & d                         & 1.04                                       & d                         & 0.93                                                  \\
\{5,11,14\}                 &                  & 0.39                      &                           & 0.42                                        &                           & 0.42                                                  \\
\{11,13,15\}                &                  & 0.32                      &                           & 0.40                                        &                           & 0.42                                                  \\
\{1,5,9,13\}                &                  & 0.69                      &                           & 0.74                                        &                           & 0.70                                                  \\
\{1,3,5,15\}                &                  & 0.48                      &                           & 0.50                                        &                           & 0.58                                                  \\
\{3,7,11,15\}               &                  & 0.51                      &                           & 0.46                                        &                           & 0.50                                                  \\
\{1,2,3,4,5\}               &                  & 0.61                      &                           & 0.67                                        &                           & 0.66                                                  \\
\{11,12,13,15\}             &                  & 0.50                      & d                          & 0.76                                        & d 
          & 1.10                                                 \\
\{1,5,9,13,17\}             &                  & 0.57                      &                           & 0.74                                        &                           & 0.61                                                  \\
\{2,4,6,10,12\}             & d                & 37.55              & d                         & 21.20                               & d                         & 124.40                                   \\
\{3,7,11,15,17\}            &                  & 0.56                      &                           & 0.57                                        &                           & 0.61                                                  \\
\{2,4,7,10,12,15,17\}       &                  & 0.83                      &                           & 1.50                                       &                           & 1.10                                                 \\
\{1,5,8,11,12,13,15\}       &                  & 0.97                      & d                         & 32.58                                    & d                         & 7204.98          \\
\{1,5,9,11,12,13,15\}       &                  & 0.82                      & d                          & 62.50                          & d                             & 6132.20 \\
\{1,3,5,7,9,11,13,15,17\}   &                  & 1.13                     &                           & 2.10                                       &                           & 1.32                                                 \\
\{2,3,4,7,10,11,12,15,17\}  &                  & 1.04                     &                           & 1.89                                       &                           & 1.21                                                 \\
\{2,4,6,10,12,14\}          &                  & n/a    &                           & n/a                       &                           & n/a                                 \\
\{6,8,10,12,14,16\}         &                  & n/a     &                           & n/a                       &                           & n/a                                 \\
\{2,4,6,8,10,12,14,16\}     &                  & n/a     &                           & n/a                      &                           & n/a                                 \\ \bottomrule
\end{tabular}
\end{table}

\subsubsection{Results}
Table~\ref{tab:results17} lists the results. For each set of terminals and notion of deadlock, we report whether a deadlock is found in column `dl' (`d' means a deadlock was found) and the execution time (s) in column `time'; `n/a' indicates a timeout.

For the network considered in our experiments, finding global deadlocks is often fast, yet for larger instances it does time out. Finding local and weak deadlocks is often slower than finding global deadlocks. Generally, finding local deadlocks is faster than finding weak deadlocks.

\subsubsection{Discussion}
Table~\ref{tab:results17} shows there are sets of terminals $M$ such that it is deadlock free for all types of deadlock, contains all different types of deadlock, but also that there are instances in which there is no global deadlock, but weak and local deadlocks are found.
Observe that the results are also consistent with theory: for all instances in which a local deadlock is found, also a weak deadlock is reported (see Lemma~\ref{prop:reachable-local-is-weak-deadlock}). The results also show that examples with no global deadlock that do exhibit local deadlock are found in practice. This is consistent with Lemma~\ref{lem:global-vs-local-deadlock}. Figure~\ref{fig:xMASLDex} shows the local deadlock found for $M=\{1,5,8,11,12,13,15\}$. 

Note that execution times increase in particular for sets of terminals that require many channels in routing. This is consistent with our expectation: if more channels and more terminals are involved, the size of the reachable state space increases, which is also likely to increase the model checking time.

\section{Conclusions}\label{section:conclusions}
We formalized three different notions of deadlock in the field of packet switching networks, namely global, local and weak deadlock. We proved that a global deadlock is also a weak deadlock, and a weak deadlock is a local deadlock. Although a local deadlock is not necessarily a weak deadlock, from a local deadlock a weak deadlock can be constructed. Hence, a network has a local deadlock if and only if it has a weak deadlock.
We showed that presence of a local or weak deadlock does not imply the existence of a global deadlock.   
Moreover, we showed how deadlocks in packet switching networks can be found using nuXmv.

\paragraph{Future work} In this paper we considered networks with deterministic routing functions. The work should be generalized to non-deterministic routing functions. Furthermore, scalability of the approach in the verification of (on-chip) interconnect networks should be evaluated further.

\bibliographystyle{splncs04}
\bibliography{paper}

\end{document}